\documentclass{article}

\usepackage{fullpage}
\usepackage{authblk}

\usepackage{amsmath}
\usepackage{subfig}
\usepackage{graphicx}
\usepackage{algorithm}
\usepackage[noend]{algpseudocode}
\usepackage{amsthm}
\usepackage{mathtools}
\usepackage[colorlinks=true, allcolors=blue]{hyperref}

\newcommand{\msgfont}[1]{\textsc{#1}}
\newcommand{\msg}[2]{$\langle$\msgfont{#1}, $#2\rangle$}
\algdef{SN}[EVENT]{Event}{EndEvent}[1]{\textbf{upon}\ #1\ \algorithmicdo}{}

\newcommand{\syscaps}{Granular Synchrony}
\newcommand{\gpscaps}{Granular Partial Synchrony}
\newcommand{\gascaps}{Granular Asynchrony}

\newcommand{\sys}{granular synchrony}
\newcommand{\gps}{granular partial synchrony}
\newcommand{\gas}{granular asynchrony}

\newtheorem{definition}{Definition}
\newtheorem{theorem}{Theorem}
\newtheorem{lemma}{Lemma}

\title{\syscaps}
\date{}

\author[1]{Neil Giridharan}
\author[2]{Ittai Abraham}
\author[1]{Natacha Crooks}
\author[3]{Kartik Nayak}
\author[4]{Ling Ren}

\affil[1]{UC Berkeley}
\affil[2]{Intel Labs}
\affil[3]{Duke University}
\affil[4]{University of Illinois Urbana-Champaign}

\begin{document}
\maketitle

\begin{abstract}
    Today's mainstream network timing models for distributed computing are synchrony, partial synchrony, and asynchrony. These models are coarse-grained and often make either too strong or too weak assumptions about the network. This paper introduces a new timing model called granular synchrony that models the network as a mixture of synchronous, partially synchronous, and asynchronous communication links. 
    The new model is not only theoretically interesting but also more representative of real-world networks. 
    It also serves as a unifying framework where current mainstream models are its special cases.
    We present necessary and sufficient conditions for solving crash and Byzantine fault-tolerant consensus in granular synchrony. 
    Interestingly, consensus among $n$ parties can be achieved against $f \geq n/2$ crash faults or $f \geq n/3$ Byzantine faults without resorting to full synchrony. 
\end{abstract}

\section{Introduction}
A fundamental aspect of any distributed computation is the \textit{timing model}. 
There are three mainstream timing models: synchrony, asynchrony, and partial synchrony.
Under synchrony, messages arrive before a known upper bound $\Delta$. Under asynchrony, messages arrive in any finite amount of time. With partial synchrony~\cite{dls}, there is an unknown but finite Global Stabilization Time (GST), and the network is asynchronous before GST and synchronous afterwards. 

The synchrony model is arguably a rosy reality: even a single message that takes longer than $\Delta$ to arrive is a violation of the synchrony model (forcing us to consider either the sender or recipient to be faulty). 
On the other hand, the asynchrony model is extremely pessimistic, making it challenging, or even impossible, to design protocols in it. 
The most well-known example may be the FLP impossibility~\cite{flp}, which states that any consensus protocol that can tolerate even a single crash fault in asynchrony must have an infinite execution. This implies that deterministic consensus in asynchrony is impossible.
The partial synchrony model tries to balance asynchrony and synchrony and has been the most widely adopted in practice so far.
But it is close to asynchrony in essence and shares the same fault tolerance bounds as (randomized) asynchronous protocols.

This paper argues that the current characterization of network timings is too coarse-grained.
We recognize the variability and heterogeneity of modern networks and propose that they should be modeled in a \textit{granular manner} via a graph consisting of a mixture of synchronous, partially synchronous, and asynchronous links.
We call the new model \emph{\sys}. 

Our new model is more than yet another theoretical construct. 
It is rooted in and motivated by our understanding and characterizations of modern distributed systems and networks.
Modern distributed systems increasingly span datacenters, be it for disaster recovery or fault isolation~\cite{swift-paxos, wpaxos, epaxos}. Within datacenters, networks are mostly synchronous~\cite{synchronous-data-center}. Spikes in message delays do occur~\cite{saksham}, but such spikes are rare and almost never happen to the entire datacenter~\cite{cloudy}. 
Across datacenters and over the Internet, networks are mostly well-behaved but are susceptible to significant fluctuations~\cite{latency-variation-internet} and adversarial attacks~\cite{bgp-hijack}.

The \sys{} timing model can serve as a unifying framework for network timing models. 
Synchrony, partial synchrony, and asynchrony are all extreme cases of it.
Outside these extreme cases, the \sys{} model is a natural intermediate between synchrony and partial synchrony (or asynchrony) and gives rise to new results that can be construed as an intermediate between fundamental results in distributed computing.  

For concreteness, we focus on the problem of fault-tolerant consensus~\cite{byzantinegenerals} in this paper.
It is well-known that under synchrony, the agreement variant of consensus can be solved in the presence of $f<n$ crash faults or $f<n/2$ Byzantine faults (assuming digital signatures). 
With partial synchrony, fewer faults can be tolerated: $f<n/2$ crash faults or $f<n/3$ Byzantine faults~\cite{dls}.
Asynchrony has the same fault thresholds and further requires the use of randomization~\cite{flp}.

We derive necessary and sufficient conditions for solving crash fault-tolerant (CFT) and Byzantine fault-tolerant (BFT) consensus in \sys.
A key benefit and interesting implication of the \sys{} model is that we do \emph{not} have to assume full synchrony to tolerate $f\geq n/2$ crash faults or $f\geq n/3$ Byzantine faults.
Instead, consensus can be reached if and only if the underlying communication graph satisfies certain conditions. 

We remark that all our protocols are \emph{graph-agnostic}, meaning they do not need to know the synchronicity property of any link.
As a result, our protocols can work in the following alternative formulation of the \sys{} model.
The consensus algorithm is parameterized by $n$ and $f$. Initially, all communication links are synchronous. 
The adversary has the power to corrupt $f$ nodes and alter some links to be partially synchronous or asynchronous but must not violate the necessary condition for the given $n$ and $f$.   
On the other hand, most of our impossibility proofs rule out algorithms that know the graph and are tailored for the graph. 
This strengthens both our protocols and our impossibility results. 

We will consider two variants of the  \sys\ model.
The first variant only has synchronous and partially synchronous links (no asynchronous links), and we refer to it as \emph{\gps}.
CFT consensus in \gps{} can be solved if and only if any quorum of $n-f$ nodes collectively can communicate synchronously with at least $f+1$ nodes despite faulty nodes. 
BFT consensus in \gps{} can be solved if and only if any set of $n-2f$ correct nodes can communicate synchronously with at least $f+1$ correct nodes despite faulty nodes.
Our CFT protocol in \gps{} relies on this condition to guarantee intersection between two quorums of size $n-f$, a crucial property for many classic consensus protocols. Without the identified condition, two quorums of size $n-f$ may not intersect. Leveraging this condition, we can expand a quorum of size $n-f$ to $f+1$ after some bounded delay. This larger quorum of size $f+1$ is guaranteed to intersect with the other quorum of size $n-f$. We use a similar argument to show that two quorums of $n-2f$ correct nodes intersect in BFT. 


The second variant further allows asynchronous links, and we refer to it as \emph{\gas}. 
For CFT consensus to be solved deterministically in \gas, it is additionally required that after removing all asynchronous edges and all crashed nodes, less than $n-f$ nodes are outside the largest connected component of the remaining graph. 
For undirected graphs, this condition is weaker than the 
correct $\diamond f$-source condition in~\cite{marcos-minimal} (see \S\ref{s:marcos-minimal}) and establishes the minimum synchrony condition needed to circumvent the FLP impossibility~\cite{flp}.  
For BFT consensus to be solved deterministically in \gas\ by a \emph{graph-agnostic} algorithm, it is additionally required that there is a correct node with partially synchronous paths to at least $f$ other correct nodes.
Our \gas{} protocols rely on these conditions to ensure that eventually a correct leader will be able to make progress without a quorum of $n-f$ nodes initiating a view change.
We leave the necessary and sufficient condition for BFT algorithms that know the graph as an open question.

\section{Model and Definitions}\label{s:defs}

We assume communication links are bi-directional. 
In \gps, each link can be either synchronous or partially synchronous.
In \gas, each link can be synchronous, partially synchronous, or asynchronous.
A synchronous link delivers each message sent on the link within a known upper bound $\Delta$.
A partially synchronous link respects the $\Delta$ message delivery bound after GST. 
An asynchronous link has no delay bound and just has to deliver each message eventually. 
We assume all communication links are reliable and FIFO (first-in-first-out), and deliver each transmitted message exactly once.

Beyond this, the model is the same as traditional consensus literature. 
There are $n$ nodes in total.
The adversary can corrupt up to $f$ nodes and can do so at any time during the protocol execution (i.e., the adversary is adaptive). 
In the CFT case, faulty nodes can fail by crashing only. In the BFT case, faulty nodes can behave arbitrarily and can be coordinated by the adversary.
For BFT, we further assume the existence of digital signatures and public-key infrastructure (PKI) and that faulty nodes cannot break cryptographic primitives. 
A message is only considered valid by correct nodes if its accompanying signature is verified (we omit writing these signature operations in the protocols).

Our protocols do not require any form of clock synchronization among nodes, and instead just require bounded clock skews.
To elaborate, certain steps of our protocols require nodes to wait for some amount of time (e.g., $4\Delta$).
For simplicity, our protocol description assumes each node will wait for precisely that amount of time.
But it is not hard to see that our protocols still work if each node waits for a time that falls in a known bounded range (e.g., between $4\Delta$ and $5\Delta$), which is easy to achieve with bounded clock skews.   

It is convenient to describe the network as an undirected graph $G=(V, E)$.
Each vertex represents a node, and each edge represents a communication link.
We use vertex and node interchangeably, and edge and link interchangeably. 
Our protocols are graph agnostic: they do not assume knowledge of the graph.

\begin{definition}[Synchronous path]
Node $a$ has a synchronous path to node $b$, written as $a \rightarrow b$, if there exist a sequence of synchronous edges $(a,i_1),(i_1,i_2),,\ldots,(i_k,b)$ where every intermediate node $i_j$ is correct.
\end{definition}

Note that in the above definition, only intermediate nodes need to be correct. Therefore, every node, even a faulty one, has a synchronous path to itself, i.e., $a \rightarrow a, \forall a \in V$. 
We generalize the notion of synchronous 
paths from two nodes to two sets of nodes $A$ and $B$.



\begin{definition}
$A \rightarrow B$ if $\forall b \in B, \exists a \in A$ such that $a \rightarrow b$.
\end{definition}



\begin{definition}[Path length, distance and diameter]
The length of a path is the number of edges in it.
If $a \rightarrow b$, the synchronous distance between these two nodes is the length of the shortest synchronous path between them. 
The synchronous diameter of a graph $G$ is
\[d(G) \coloneqq \max_{\substack{F,a,b~\text{s.t.}~|F| \leq f,~a\rightarrow b}} d(a,b).\] 
\end{definition}

Partially synchronous path, path length, distance, and diameter $d'(G)$ are similarly defined.
Note that a partially synchronous path can contain synchronous edges. 

The (partially) synchronous distance is only defined for a pair of nodes that have a (partially) synchronous path between them. 
We also remark that for the Byzantine case, distance is only defined for a pair of correct nodes. 
The $\max$ in the diameter definition is taken over all pairs with the corresponding distance defined. 
The two diameters capture the worst-case round-trip delays among nodes connected by synchronous and partially synchronous paths, respectively. 
If $d(G)$ or $d'(G)$ is known, they can be directly used in our protocols; otherwise, $|V|-1$ is a trivial upper bound. 
We will simply write $d$ and $d'$ when there is no ambiguity.  

\begin{definition}[Consensus] In a consensus protocol, every node has an initial input value and must decide a value that satisfies the following properties.   
\begin{itemize}
    \item \textit{Agreement:} No two correct nodes decide different values.\footnote{For CFT consensus, we actually achieve the stronger property of uniform agreement, which states that no two nodes (even faulty ones) decide differently.}
    \item \textit{Termination:} Every correct node eventually decides.
    \item \textit{Validity:} If all nodes have the same input value, then that is the decision value.
    %
\end{itemize}
\label{def:consensus}
\end{definition}

\section{CFT Consensus in \gpscaps}

\begin{theorem}\label{thm:cft-gps}
Under \gps, CFT consensus on a graph $G=(V, E)$ is solvable if and only if, regardless of which up to $f$ nodes are faulty,
$\forall A \subseteq V$ with $|A| \geq n-f$, $\exists B \subseteq V$ with $|B| \geq f+1$ such that $A \rightarrow B$. 
\end{theorem}

In words, the condition is that any set $A$ of size at least $n-f$ has a potentially larger set $B$ of size at least $f+1$, such that for any node $b \in B$ there exits $a\in A$ and a synchronous path from $a$ to $b$. Intuitively, if a message arrives at all of $A$, then it will arrive at all of $B$ after some delay.

It is worth noting that classic crash fault tolerance bounds are special cases of our theorem. 
For example, when all links are synchronous, any node has synchronous paths to all $n$ nodes. 
Thus, synchronous CFT consensus can be solved for any $n \geq f+1$. 
At the other extreme, $n=2f+1$ is the smallest value of $n$ for which the condition in Theorem~\ref{thm:cft-gps} trivially holds even when all edges are partially synchronous (see necessity proof). 
The more interesting part of our theorem is of course when we have a mix of synchronous and partially synchronous edges.
Figure \ref{fig:cft-necessary-figure} gives examples of these intermediate cases where CFT consensus is solvable with $f+1<n\leq 2f$. 

\begin{figure}
    \centering
    \subfloat[]{\includegraphics[width=2cm,height=3cm,keepaspectratio]{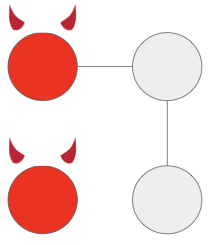}}
    \qquad
    \subfloat[]{\includegraphics[width=2cm,height=3cm,keepaspectratio]{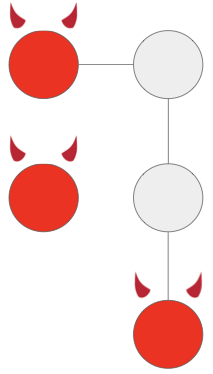}}
    \qquad
    \subfloat[]{\includegraphics[width=2cm,height=3cm,keepaspectratio]{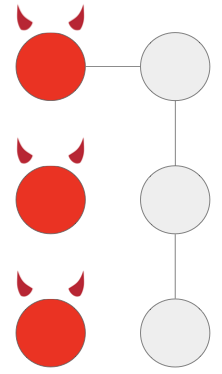}}
    \qquad

    \caption{Only synchronous links are shown in the figure for brevity. Faulty nodes are denoted in red with horns, and the correct nodes are denoted in gray. 
    The figure shows the necessary and sufficient condition in theorem \ref{thm:cft-gps} being satisfied for
    \textbf{(a)} $n=4$, $f=2$,  \textbf{(b)} $n=5$, $f=3$, and \textbf{(c)} $n=6$, $f=3$.}
    \label{fig:cft-necessary-figure}
\end{figure}

\subsection{Necessity}\label{s:cft-gps}

We first prove the ``only if'' part of Theorem~\ref{thm:cft-gps}. The proof is similar to the DLS proof in partial synchrony~\cite{dls}. 
To ensure agreement, we must ensure that nodes cannot be partitioned into two disjoint groups with no synchronous inter-group links. The condition in Theorem~\ref{thm:cft-gps} ensures exactly that.

\begin{proof}
For $n\geq 2f+1$, the ``only if'' part of the theorem is vacuous because the condition trivially holds: $n-f \geq f+1$, and every node has a synchronous path to itself.

For $n\leq 2f$, we prove by contradiction. 
Suppose there is an algorithm that solves consensus on a graph $G$ that does not satisfy the condition in the theorem. 
Then, there exists a set $F$ of up to $f$ nodes such that, if nodes in $F$ crash, there exists a set $A$ of at least $n-f$ nodes, which collectively have synchronous paths to at most $f$ nodes. 
Let $B$ be the set of these $f$ nodes excluding $A$. 
Let $C$ be the remaining nodes, i.e., $C=[n] \setminus \{A \cup B\}$. 
Note that $\{A,B,C\}$ is a three-way disjoint partition of the $n$ nodes.
Also note that $|A \cup B| \leq f$ and $|B \cup C| = n-|A| \leq f$.
Next, we consider three executions. 

In execution 1, all nodes have input $v_1$ and 
nodes in $B \cup C$ crash at the beginning.
Since $|B \cup C| \leq f$, $A$ eventually decides $v_1$ in time $t_1$ due to validity.
In execution 2, all nodes have input $v_2 \neq v_1$ and 
nodes in $A \cup B$ crash at the beginning. 
Since $|A \cup B| \leq f$, $C$ eventually decides $v_2$ in time $t_2$ due to validity.

In execution 3, nodes in $A$ have input $v_1$, nodes in $C$ have input $v_2$, nodes in $B$ crash at the beginning, and $GST > \max(t_1, t_2)$.
Note that crashing $B$ (instead of $F$) does not change the fact that $A$ has synchronous paths to $A \cup B$ only.
This is because, with $B$ crashed, nodes in $F \setminus B$ do not have synchronous paths to $A$ themselves (otherwise, they would have synchronous paths to $A$ with $F$ crashed).
Thus, synchronous paths from $A$ to $C$ cannot go through $F \setminus B$. 
Because there are no synchronous edges between $A$ and $C$, the adversary can delay the delivery of all messages between $A$ and $C$ until after GST.
Thus, $A$ cannot distinguish execution 3 from execution 1 and $C$ cannot distinguish execution 3 from execution 2. 
Then, $A$ decides $v_1$ and $C$ decides $v_2$, violating agreement.
\end{proof}

\subsection{Protocol}

Next, we present a new CFT consensus protocol assuming the condition in theorem~\ref{thm:cft-gps} holds. This establishes the sufficiency of the condition. 

\par \underline{\textit{Overview.}} A natural starting point is a standard quorum-based partially synchronous CFT consensus protocol.
Such protocols require $n>2f$ to ensure any two quorums of size $n-f$ intersect. 
When $n \leq 2f$, two quorums of $n-f$ may not intersect. 
But when the condition in theorem \ref{thm:cft-gps} holds, a quorum of $n-f$ nodes can hear from $f+1$ nodes of any critical information in bounded time.
This effectively promotes a quorum of size $n-f$ to $f+1$ and ensures safety as a quorum of size $f+1$ always intersects a quorum of size $n-f$. 

Similar to other leader-based partially synchronous consensus protocols, our protocol operates in a series of views, where each view has a leader. The leader of view $v$ is denoted as $L_v$. 
Leaders can be elected using a simple round-robin order. If a view after GST has a correct leader, nodes will commit that leader's proposal and terminate. There is a view change procedure to replace a leader who is not making progress. We focus on a single-shot consensus here, but the protocol can be easily adapted to the multi-shot setting.

\par \underline{\textit{Locks.}} A $lock\coloneqq (view, value)$ consists of a view and value. Initially, each node locks on its input value with view number $0$.
When a node receives a proposal from the leader of the current view, it updates its lock to the current view and the proposed value. 
Locks are ranked by view numbers. Note that except for the initial view $0$, there cannot be two locks with the same view number but different values, since only one value is proposed per view. 
Locks from view 0 can be ranked arbitrarily. 

\begin{algorithm}[tb]
\caption{CFT consensus protocol in \gps{} for node $i$}\label{alg:gps-cft}
\begin{algorithmic}[1]
\State{$v_i \gets 0$}\Comment{Initialize local view number}
\State{$lock \gets (0, input_i)$ }    \Comment{Initially lock on the input value}
\State{\textbf{enter} view 1}    

\medskip
\Event{entering view $v$}
\State{$v_i \gets v$}
\State{\textbf{start} $view\_timer \gets timer(4\Delta)$}\Comment{Timer for changing view}
\State{\textbf{send} \msg{Status}{v,lock} to $L_v$}\label{send-status}
\EndEvent

\medskip
\Event{receiving $n-f$ \msg{Status}{v_i,-} and $i=L_{v_i}$}
\State{$val \gets$ value from the highest lock (by view) received}\label{highest-rank}
\State{\textbf{send} \msg{Propose}{v_i,val} to all}\label{propose-send}\Comment{Leader proposal}
\EndEvent


\medskip
\Event{receiving \msg{Propose}{v_i,val}}
\State{$lock \gets (v_i,val)$}\label{update-lock} 
\State{\textbf{send} \msg{Vote}{v_i,val} to all}\label{vote-send}
\EndEvent


\medskip
\Event{receiving $n-f$ \msg{Vote}{v_i,val} or \msg{Commit}{val}}
\State{\textbf{send} \msg{Commit}{val} to all}\label{send-commit}
\State{\textbf{commit} $val$ and \textbf{terminate}} 
\EndEvent



\medskip
\Event{$view\_timer$ expiring}
\State{\textbf{send} \msg{NewView}{v_i+1} to all}\label{blame}
\EndEvent

\medskip
\Event{receiving \msg{NewView}{v} where $v>v_i$}
\State{\textbf{echo} \msg{NewView}{v} and to all}
\State{\textbf{send} \msg{Locked}{lock} to all}
\State{\textbf{stop} accepting \textsc{Propose} messages in views up to $v-1$}
\State{\textbf{wait} $2d\Delta$ time}
\State{\textbf{enter} view $v$}
\EndEvent

\medskip
\Event{receiving \msg{Locked}{lock'}}
\State{$lock\gets $ higher lock (by view) between $lock$ and $lock'$}
\State{\textbf{echo} \msg{Locked}{lock'} to all}
\EndEvent

\end{algorithmic}
\end{algorithm}

We describe the protocol next.

\par \underline{\textit{Status step.}}
Each view begins with every node sending a \textsc{Status} message to the leader of the current view. A node also starts a timer for the view. 

\par \underline{\textit{Leader proposal step.}} When $L_v$ is in view $v$ and receives $n-f$ \msg{Status}{v,-} messages, it proposes the highest locked value among those. Note that $L_v$ only sends one \textsc{Propose} message in a view.
When a node is in view $v$ and receives a \msg{Propose}{v,val} message, it updates its lock to $(v,val)$ and sends a \msg{Vote}{v,val} message to all nodes.


\par \underline{\textit{Commit step.}} When a node receives a quorum of $n-f$ \msg{Vote}{v,val} messages or a single \msg{Commit}{val} message, it commits $val$, sends a \msg{Commit}{val} message to all nodes, and terminates.

\par \underline{\textit{View change step.}} When a node times out in a view $v$ without committing a value, it sends \msg{NewView}{v+1} to all nodes, asking them to move to the next view.
%
Upon receiving \msg{NewView}{v} for a higher view $v$, a node echoes \msg{NewView}{v} and its own lock to all nodes, waits for $2d\Delta$ time, and then enters view $v$. 
During this waiting period, the node will not send \textsc{Vote} for its current view but will listen for \textsc{Locked} messages to update its lock and also echo locks. 
The $2d\Delta$ time accounts for the worst-case round-trip delay to send a \textsc{NewView} message and receive the \textsc{Locked} message.

\subsection{Analysis}




\begin{lemma}\label{lem:gps-cft-3}
If some node commits $val$ in view $v$, then any \msg{Propose}{v',val'} message in view $v'\geq v$ must have $val'=val$.
\end{lemma}

\begin{proof}
We prove this lemma by induction on view $v'$. 
The base case of $v'=v$ is straightforward since each leader proposes only one value, so $val'=val$. 

For the inductive step, suppose the lemma holds up to view $v'-1$, and we consider view $v'$. Suppose for the sake of contradiction that some node commits $val$ in view $v$, and there is a \msg{Propose}{v',val'} message from $L_{v'}$ for $val'\neq val$. 
$L_{v'}$ must have received \msg{Status}{v',-} messages from a set $P$ of $n-f$ nodes.
By the condition in theorem~\ref{thm:cft-gps}, $P\rightarrow Q$, where $Q$ is a set of $f+1$ nodes. 
Since a node committed $val$ in view $v$, there must exist a set $R$ of $n-f$ nodes that sent \msg{Vote}{v,val} messages and updated $lock\coloneqq (v,val)$ in view $v$. Sets $Q$ and $R$ intersect in at least one node. Let this node be $q$. 


Since the graph is undirected, there must exist a node $p\in P$ such that $q\rightarrow p$. By the induction hypothesis, \textsc{Propose} messages from view $v$ to $v'-1$ must be for $val$. Since a node only updates its lock monotonically based on view numbers, node $q$ must have a lock with view $\geq v$ for $val$. 
Let $t_p$ be the time node $p$ echoed \msg{NewView}{v'}. By time $t_p+d\Delta$, node $q$ receives \msg{NewView}{v'}. Upon receiving \msg{NewView}{v'}, node $q$ sends a \msg{Locked}{lock} message to all nodes. This lock is received by node $p$ by time $t_p+2d\Delta$. Node $p$ updates its lock to view $\geq v$ for $val$ before entering view $v'$. Thus, $L_v$ receives at least one \textsc{Status} message for $val$ with view $\geq v$ and propose $val$, a contradiction.   
\end{proof}



\begin{theorem}[Agreement]
No two nodes commit different values.    
\end{theorem}

\begin{proof}
Let $v$ be the smallest view in which a node commits some value, say $val$. 
Since only $val$ can be proposed in view $v$ and all subsequent views by lemma~\ref{lem:gps-cft-3}, no node can commit a different value. 
\end{proof}

\begin{theorem}[Termination]
All correct nodes eventually decide.    
\end{theorem}

\begin{proof}
With round-robin leader election, correct nodes are elected leaders infinitely often. Thus, there must be a view $v$, after $GST+2d\Delta$, whose leader is correct. We next prove that all nodes will decide and terminate in view $v$ (if they don't decide earlier). 

Let $t$ ($t \geq GST+2d\Delta$) be the first time some correct node enters view $v$.
This correct node sends \msg{NewView}{v} to all nodes at $t-2d\Delta \geq GST$. 
All correct nodes receive \msg{NewView}{v} by time $t-2d\Delta+\Delta$, wait $2d\Delta$ themselves, and enter view $v$ by time $t+\Delta$. 
Upon entering view $v$, they send \msg{Status}{v,-} messages to $L_v$. $L_v$ receives $n-f$ \msg{Status}{v,-} messages by time $t+2\Delta$, and sends a \msg{Propose}{v,-} message to all nodes. 
All correct nodes receive the \msg{Propose}{v,-} message and send \msg{Vote}{v,-} messages by time $t+3\Delta$. 
All correct nodes receive $n-f$ \msg{Vote}{v,-} messages and commit by time $t+4\Delta$. 
Since a node's view timer is $4\Delta$, all correct nodes commit and terminate in view $v$.
\end{proof}

\begin{theorem}[Validity]
If all nodes have the same input $val$, then all correct nodes eventually decide $val$.  
\end{theorem}
\begin{proof}
If all nodes have the same input $val$, all nodes set $lock \gets (0, val)$. Following a similar proof as in lemma~\ref{lem:gps-cft-3}, no other value can be proposed in all subsequent views. Validity follows from termination. 
\end{proof}


\section{CFT Consensus in \gascaps}

\begin{theorem}\label{thm:cft-gpa}
Under \gas, CFT consensus on a graph $G=(V,E)$ can be solved deterministically if and only if, (i) the condition in theorem~\ref{thm:cft-gps} holds and (ii) for all $F$ with $|F| \leq f$, less than $n-f$ nodes are outside the largest connected component of $G'=(V-F,\diamond E)$ where $\diamond E$ is the set of synchronous and partially synchronous edges. 
\end{theorem}

In other words, condition (ii) says that if we remove all asynchronous edges and all faulty nodes from $G$ and further remove the largest connected component in the remaining graph, then there are fewer than $n-f$ nodes left.

\subsection{Necessity}

\begin{proof}
Condition (i) is already proved to be necessary in theorem~\ref{thm:cft-gps}. We focus on condition (ii). 
Suppose for the sake of contradiction there exists a deterministic algorithm $\mathcal{A}$ that solves CFT consensus on a graph $G$ that violates condition (ii). 
This means there exists a set $F$ with $|F|\leq f$ such that removing the largest connected component from $G'=(V-F,\diamond E)$ ($G$ with $F$ and all asynchronous edges removed) leaves $\geq n-f$ nodes.

Suppose the graph $G'$ has $q$ connected components. Clearly, $q > 1$.
Let $C_i$ be $i$-th connected component in $G'$.
We have $|F \cup C_i| \leq f$ for all $i$ because even the largest connected component plus $F$ has at most $f$ nodes. 


We construct an external system consisting of $q$ nodes connected only by asynchronous links. 
We can convert $\mathcal{A}$ into a deterministic algorithm that solves consensus in this external system while tolerating one crash fault. 
To do so, let the $i$-th node in the external system, $q_i$, simulate the nodes in $C_i$ in $\mathcal{A}$. 
If $q_i$ has input $v_i$, then all nodes in $C_i$ have input $v_i$ in the simulation. 

An execution in this external system with $q_i$ crashing at time $t$ faithfully simulates an execution of $\mathcal{A}$ with $F$ crashing in the beginning and $C_i$ crashing at time $t$. 
In particular, observe that two connected components in $G'$ only have asynchronous edges between them once nodes in $F$ crash.  
Since $|F \cup C_i| \leq f$ for all $i$, $\mathcal{A}$ solves consensus in the original system.
Thus, the simulated algorithm solves consensus deterministically in the external system while tolerating one crash fault in asynchrony.
This contradicts the FLP impossibility~\cite{flp}. 
\end{proof}



\subsection{Protocol}

Next, we adapt our previous CFT consensus protocol in algorithm~\ref{alg:gps-cft} from \gps{} to \gas{}, assuming the condition in theorem~\ref{thm:cft-gpa} holds. This establishes the sufficiency of the condition.

Our prior CFT consensus protocol still maintains safety under \gas{}, but liveness no longer holds because there is no time when all edges behave synchronously (asynchronous links do not have a $GST$ assumption). 
As a result, correct leaders in our prior protocol may continuously time out.
Luckily, condition (ii) in theorem~\ref{thm:cft-gpa} can be leveraged to guarantee that when the set $F$ of crashed nodes stops growing, and a correct node in the largest connected component of $G'=(V-F,\diamond E)$ is elected leader after GST, this leader will not be replaced and will make progress. 
To do so, we first require $n-f$ nodes to initiate a view change.
This way, because all nodes in $F$ are crashed and fewer than $n-f$ nodes are outside the largest connected component of $G'=(V-F,\diamond E)$, we just need to make sure that no node in this largest connected component initiates a view change. 
This technique is similar to those used in view synchronizers~\cite{alexey-live,alexey-live-2} to make sure correct nodes eventually overlap and remain in the same view to ensure termination.



\begin{algorithm}[tb]
\caption{CFT consensus protocol in \gas{} for node $i$}\label{alg:gas-cft}
\begin{algorithmic}[1]
\State{$v_i \gets 0$}\Comment{Initialize local view number}
\State{$lock \gets (0, input_i)$ }    \Comment{Initially lock on the input value}
\State{\textbf{enter} view 1}    

\medskip
\Event{entering view $v$}
\State{$v_i \gets v$}
\State{\textbf{send} \msg{Status}{v,lock} to all}
\EndEvent


\medskip
\Event{receiving $n-f$ \msg{Status}{v_i,-} where $i \neq L_{v_i}$}
\State{\textbf{echo} these $n-f$ \msg{Status}{v_i,-} to all}
\State{\textbf{start} $proposal\_timer\gets timer(3d'\Delta)$}
\EndEvent


\medskip
\Event{receiving \msg{Propose}{v_i,val}}
\State{$lock\gets (v_i,val)$}
\State{\textbf{echo} \msg{Propose}{v_i,val} to all}
\State{\textbf{send} \msg{Vote}{v_i,val} to all}
\EndEvent

\medskip
\Event{$proposal\_timer$ expiring and no leader proposal received}
\State{\textbf{send} \msg{ViewChange}{v_i} to all} \label{blame-gas}
\EndEvent

\medskip
\Event{receiving $n-f$ \msg{ViewChange}{v}}
\State{\textbf{send} \msg{NewView}{v+1} to all} \label{blame-cert-gas}
\EndEvent


\medskip
\State{\textsc{Vote, Commit, Locked, NewView} messages at all nodes and \textsc{Status} messages at view leaders are processed the same way as in Algorithm~\ref{alg:gps-cft}}

\end{algorithmic}
\end{algorithm}

We only describe the status and view change steps since the rest of the protocol remains the same as algorithm~\ref{alg:gps-cft}.

\par \underline{\textit{Status and propose step.}}
Upon entering a new view $v$, a node sends a \msg{Status}{v,lock} message to \textbf{all} nodes. When a node receives at least $n-f$ \msg{Status}{v,-} messages, it forwards this set of \textsc{Status} messages to all nodes and starts a timer of $3d'\Delta$ duration. 
Upon receiving a proposal, a node forwards the proposal to all nodes, in addition to locking on and voting for the proposal.
The same vote and commit steps from algorithm~\ref{alg:gps-cft} follow.

\par \underline{\textit{View change.}} A node suspects the leader is faulty if it does not receive a \msg{Propose}{v,-} message before its timer expires. When this occurs, a node sends a \msg{ViewChange}{v} message to all nodes, indicating it wishes to quit view $v$. When a node receives $n-f$ \msg{ViewChange}{v} messages for the current view $v$, it sends a \msg{NewView}{v+1} message to all nodes. Upon receiving a \textsc{NewView} message, a node carries out the same new view step from algorithm \ref{alg:gps-cft}.

\subsection{Analysis}

The agreement and validity proofs are identical to the \gps{} CFT case. We focus on termination.



\begin{lemma}
If no correct node ever terminates, then every correct node keeps entering higher views.
\label{lemma:cft-gas-liveness}
\end{lemma}

\begin{proof}
Suppose for the sake of contradiction, there exists a correct node $n_1$, which never enters a higher view. Let $v$ be the view $n_1$ is in.
If any correct node ever enters a view higher than $v$, it sends a \textsc{NewView} message for that higher view to all nodes. $n_1$ will eventually receive this higher \textsc{NewView} message and enter a higher view, a contradiction. 
Thus, no node ever enters a view higher than $v$.
Before entering view $v$, $n_1$ has sent \msg{NewView}{v} to all nodes. All correct nodes will eventually receive this \msg{NewView}{v} message, enter view $v$, and send \msg{Status}{v,-} messages. Eventually, correct nodes will receive $n-f$ \msg{Status}{v,-} messages and start their proposal timers. If $n_1$ receives $n-f$ \msg{ViewChange}{v} messages, it will enter view $v+1$, a contradiction. Thus $n_1$ never receives $n-f$ \msg{ViewChange}{v} messages. 
Then, there must be at least one correct node that never sends \msg{ViewChange}{v} and instead echoes \msg{Propose}{v,-} to all nodes. Eventually, all correct nodes will receive \msg{Propose}{v,-} message and send \msg{Vote}{v,-} messages to all nodes. Eventually $n_1$ will receive $n-f$ \msg{Vote}{v,-} messages and terminate, a contradiction.
\end{proof}

\begin{theorem}
All correct nodes eventually terminate.    
\end{theorem}

\begin{proof}
Suppose for the sake of contradiction that some correct node never terminates.
Observe that if one correct node terminates, it sends a \textsc{Commit} message and makes all correct nodes eventually terminate. 
Thus, no correct node ever terminates. 
By lemma~\ref{lemma:cft-gas-liveness}, every correct node keeps entering higher views.

Eventually, there will be a first time after $GST+2d\Delta$ that some correct node enters a view $v$ such that (i) the set $F$ of crashed nodes no longer grows in views $\geq v$, (ii) $L_v \not\in F$, and (iii) $L_v$ is in the largest connected component $G'=(V-F,\diamond E)$.
Let $C$ denote this largest connected component. 
We next prove no node in $C$ will ever send \msg{ViewChange}{v}.

Let $p$ be the first node in $C$ that enters view $v$, and let $p$ enter view $v$ at time $t > GST+2d\Delta$. 
Observe that no node in $C$ will send \msg{ViewChange}{v} before time $t+3d'\Delta$ (proposal timer duration is $3d'\Delta$). 
Nodes in $F$ crashed before entering view $v$ and cannot send \msg{ViewChange}{v}.
Due to the condition in theorem~\ref{thm:cft-gpa}, $n-|C \cup F|<n-f$.
Thus, there will not be $n-f$ \msg{ViewChange}{v} messages before $t+3d'\Delta$. 

$p$ sends \msg{NewView}{v} at time $t-2d\Delta>GST$. 
All nodes in $C$ receive \msg{NewView}{v} by $t-2d\Delta+d'\Delta$, enter view $v$ by $t+d'\Delta$, and stay in view $v$ at least until $t+3d'\Delta$.

When a node $q \in C$ receives $n-f$ \msg{Status}{v,-} messages at time $t'>t$, $q$ echoes these $n-f$ messages and starts its proposal timer. 
All nodes in $C$ enter view $v$ by time $t+d'\Delta$ and are ready to echo these \msg{Status}{v,-} messages.
(Recall that $d'$ is the partially synchronous diameter of the graph.)
$L_v$, which is in $C$, receives these $n-f$ \msg{Status}{v,-} messages by time $\max(t+2d'\Delta, t'+d'\Delta)<t'+2d'\Delta$.
$L_v$ sends a \msg{Propose}{v,-} message by time $t'+2d'\Delta$ and it reaches $q$ by time $t'+3d'\Delta$, which is before $q$'s proposal timer expires. 
Thus, $q$ does not send \msg{ViewChange}{v}. 
This establishes that no node in $C$ will ever send \msg{ViewChange}{v}.
Again, nodes in $F$ never send \msg{ViewChange}{v}.
Since $n-|C \cup F|<n-f$, there will never be $n-f$ \msg{ViewChange}{v} messages.
Thus, no correct node ever enters a view higher than $v$.
This contradicts lemma~\ref{lemma:cft-gas-liveness}.
\end{proof}
\section{BFT Consensus in \gpscaps}

\begin{theorem}\label{thm:bft-gps}
Under \gps, BFT consensus with $n \geq 2f+1$ on a graph $G$ is solvable if and only if, for any set $F$ of at most $f$ faulty nodes,
$\forall A \subseteq V - F$ with $|A| \geq n-2f$, $\exists B \subseteq V-F$ with $|B| \geq f+1$ such that $A \rightarrow B$. 
\end{theorem}

In words, the condition is that any honest set $A$ of size at least $n-2f$ has a potentially larger honest set $B$ of size at least $f+1$, such that for any node $b \in B$ there exits $a\in A$ and a synchronous path from $a$ to $b$. Intuitively, if a message arrives at all of $A$, then it will also arrive at all of $B$ after some delay.

Note that in BFT consensus, it never hurts the adversary to corrupt the maximum number of nodes allowed since Byzantine nodes can actively participate. This is why we can focus on the case of $|F|=f$ (as opposed to $|F| \leq f$). 

Observe that the classic Byzantine fault tolerance bounds are special cases of our theorem. 
For example, when $n=2f+1$ and all links are synchronous, any $n-2f=1$ correct node has synchronous paths to all $n-f=f+1$ correct nodes, so consensus is solvable. 
At the other extreme, $n=3f+1$ is the smallest value of $n$ for which the condition in theorem \ref{thm:bft-gps} trivially holds even when all edges are partially synchronous (see necessity proof).
And again, we will focus on the more interesting region of $2f+1 < n \leq 3f$.

\subsection{Necessary}\label{s:hps-bft-necessary}

The proof is again very similar to DLS~\cite{dls}.
The essence of the condition (and the proof) is to prevent a ``split-brain'' attack in which two groups of $n-2f$ correct nodes cannot communicate in time and separately make progress with $f$ Byzantine nodes. 

\begin{proof}[\textbf{Proof of Theorem~\ref{thm:bft-gps} necessity part}]
For $n\geq 3f+1$, the theorem is vacuous because the condition trivially holds: any set of $n-2f \geq f+1$ correct nodes have synchronous paths to at least $f+1$ correct nodes (i.e., themselves).

For $n\leq 3f$, we prove by contradiction. 
Suppose there is an algorithm that solves consensus on a graph $G$ that does not satisfy the condition in the theorem. 
Then, there exists a set $F$ of $f$ nodes such that, if nodes in $F$ are faulty, a set $A$ of $n-2f$ correct nodes collectively have synchronous paths to at most $f$ correct nodes. 
Let $B$ be the set of these $f$ nodes excluding $A$.
Let $C$ be the remaining nodes, i.e., $C=[n] \setminus \{F \cup A \cup B\}$. 
Note that $\{A,B,F,C\}$ is a four-way disjoint partition of the $n$ nodes.
Also note that $n-2f = |A| \leq |A \cup B| \leq f$, $|F| = f$, and $|C| = n-|F \cup A\cup B| \leq f$. 

Next, we consider three executions.
In execution 1, all nodes have input $v_1$, and nodes in $C$ are Byzantine. 
Since $|C| \leq f$, $A\cup B$ eventually decide $v_1$ in time $t_1$ due to validity.
In execution 2, all nodes have input $v_2$, and nodes in $A \cup B$ are Byzantine. Since $|A \cup B| \leq f$, $C$ eventually decide $v_2$ in time $t_2$ due to validity.

In execution 3, nodes in $A \cup B$ have input $v_1$, nodes in $C$ have input $v_2$, nodes in $F$ are Byzantine, and $GST > \max(t_1, t_2)$. 
$F$ will behave towards $A\cup B$ like in execution 1 and towards $C$ like in execution 2. 
Because there is no synchronous link between $A\cup B$ and $C$, $A \cup B$ cannot distinguish execution 3 from execution 1 and $C$ cannot distinguish execution 3 from execution 2. 
Thus, $A \cup B$ decides $v_1$ and $C$ decides $v_2$, violating agreement.
\end{proof}

\subsection{Protocol}

Next, we give a new BFT consensus protocol assuming the condition in theorem~\ref{thm:bft-gps} holds. 
The protocol we present here achieves external validity~\cite{cachin-aba}.
In appendix~\ref{s:unanimity}, we show how to extend it to achieve the strong unanimity validity in definition~\ref{def:consensus}.
This establishes the sufficiency of the condition. 

Like in the CFT case, we will start from a standard leader-based partially synchronous BFT protocol and then take advantage of our graph condition to upgrade a quorum of $n-2f$ correct nodes to $f+1$ correct nodes.  

A $lock$ is a set $L$ of $n-f$ signed matching \msg{Vote-1}{view,val} messages from distinct nodes. 
Locks are ranked by their view numbers. 
%
%
We describe the protocol next.
\par \underline{\textit{Status step.}} Each view begins with every node sending a \textsc{Status} message to the leader of the current view. A node also starts a timer for the view.

\par \underline{\textit{Leader proposal step.}} When the leader of view $v$, $L_v$, receives a set $S$ of $n-f$ \msg{Status}{v,-} messages from distinct nodes, it picks the highest-ranked lock among those. 
If no lock is reported, then the leader can safely propose its own input value, $val_i$. Otherwise, the leader must propose the value in the highest-ranked lock. 
The leader sends \msg{Propose}{v,val,S} to all nodes. Note that a correct leader only sends one \textsc{Propose} message in a view.

\par \underline{\textit{Equivocation check step.}} When a node receives \msg{Propose}{v,val,S}, it checks whether $val$ is the highest-ranked locked value from the set $S$. If so, it forwards the \textsc{Propose} message to all nodes and starts a timer for $d\Delta$ to listen for conflicting \textsc{Propose} messages in the same view. 
If it receives a conflicting \textsc{Propose} message, it detects the leader is faulty, forwards the equivocation to all nodes, and sends a \textsc{ViewChange} message for the current view.
If the timer expires and no conflicting \textsc{Propose} message is received, the node will send a \msg{Vote-1}{v,val} message to all nodes indicating its support for the leader's proposal. 

\par \underline{\textit{Locking step.}} When a node receives $n-f$ \msg{Vote-1}{v,val} messages, it forms a lock certificate $L$ for $val$ in view $v$. 
The node updates its $lock\coloneqq L$ and sends a \msg{Vote-2}{v,val} message to all nodes. 
The equivocation check guarantees the uniqueness of the locked value in each view. 

\par \underline{\textit{Commit step.}} Upon receiving $C\gets n-f$ \msg{Vote-2}{v,val} messages, a node sends a \msg{Commit}{C} message. Upon receiving a \msg{Commit}{C} message, it commits and terminates.

\par \underline{\textit{View Change.}} 
A node sends \msg{ViewChange}{v} if it detects equivocation or times out in view $v$.
Upon receiving $f+1$ \textsc{ViewChange} messages, a node stops sending \textsc{Vote-1}/\textsc{Vote-2} messages in view $v$ and sends its lock to all nodes. A node cannot immediately enter the next view but instead must wait $2d\Delta$ time before doing so. This is to give enough time for locks to propagate in the network.

\begin{algorithm}[tbp]
\caption{BFT consensus protocol in \gps{} for node $i$}\label{alg:gps-bft}
\begin{algorithmic}[1]

\State{$v_i \gets 0,~lock \gets \bot$}\Comment{Initialize local view number and lock}

\State{\textbf{enter} view 1}  

\medskip
\Event{entering view $v$}
\State{$v_i \gets v$}
\State{\textbf{start} $view\_timer \gets timer((5+d)\Delta)$}\Comment{Timer for changing view}
\State{\textbf{send} \msg{Status}{v,lock} to $L_v$}
\EndEvent

\medskip
\Event{receiving $S \gets n-f$ \msg{Status}{v_i,-}}
\State{$val\gets$ value in the highest lock in $S$, or $input_i$ if all locks in $S$ are $\bot$}
\State{\textbf{send} \msg{Propose}{v_i,val,S} to all}
\EndEvent

\medskip
\Event{receiving \msg{Propose}{v_i,val,S} from $L_{v_i}$}
\If{$val$ matches the highest locked value in $S$ or all locks in $S$ are $\bot$}
\State{\textbf{echo} \msg{Propose}{v_i,val,S} to all} 
\State{\textbf{start} $vote\_timer \gets timer(d\Delta)$}\Comment{To detect equivocation}
\EndIf
\EndEvent

\medskip
\Event{$vote\_timer$ expiring and no equivocation detected}
\State{\textbf{send} \msg{Vote-1}{v_i,val} to all}
\EndEvent

\medskip
\Event{receiving $L\gets n-f$ \msg{Vote-1}{v_i,val}}
\State{$lock\gets L$}
\State{\textbf{send} \msg{Vote-2}{v_i,val} to all}
\EndEvent

\medskip
\Event{receiving $C\gets n-f$ \msg{Vote-2}{v_i,val} or one \msg{Commit}{C}}
\State{\textbf{send} \msg{Commit}{C} to all}
\State{\textbf{commit} $val$ and \textbf{terminate}} 
\EndEvent

\medskip
\Event{receiving \msg{Propose}{v_i,val,-} and \msg{Propose}{v_i,val',-} where $val'\neq val$}
\State{\textbf{echo} \msg{Propose}{v_i,val,-} and \msg{Propose}{v_i,val',-} to all}
\State{\textbf{send} \msg{ViewChange}{v_i} to all} 
\EndEvent

\medskip
\Event{$view\_timer$ expiring}
\State{\textbf{send} \msg{ViewChange}{v_i} to all}
\EndEvent

\medskip
\Event{receiving $VC \gets f+1$ \msg{ViewChange}{v} where $v>v_i$}
\State{\textbf{stop} sending \textsc{Vote-1}/\textsc{Vote-2} messages for views up to $v$} 
\State{\textbf{echo} $VC$ to all}
\State{\textbf{echo} \msg{Locked}{lock} to all}
\State{\textbf{wait} $2d\Delta$}
\State{\textbf{enter} view $v+1$}
\EndEvent

\medskip
\Event{receiving \msg{Locked}{lock'}}
\State{$lock\gets $ higher lock between $lock$ and $lock'$}
\State{\textbf{echo} \msg{Locked}{lock'} to all}
\EndEvent


\end{algorithmic}
\end{algorithm}

\subsection{Analysis}\label{bft-gps-analysis}

External validity is easily ensured if all correct nodes validate the proposed value before voting for it. In appendix~\ref{s:unanimity}, we show how to achieve the strong unanimity validity in definition~\ref{def:consensus}. 
We now focus on agreement and termination. 

\begin{lemma}\label{lem:hps-bft-agree-1}
If there exist $n-f$ \msg{Vote-1}{v,val} messages and $n-f$ \msg{Vote-1}{v,val'} messages in the same view $v$, then $val=val'$.
\end{lemma}

\begin{proof}
Suppose for the sake of contradiction there exist a set $S$ of $n-f$ \msg{Vote-1}{v,val} messages and a set $S'$ of $n-f$ \msg{Vote-1}{v,val'} messages where $val\neq val'$. 

Of the $n-f$ nodes whose \textsc{Vote-1} messages are in $S$, at least a set $P$ of $n-2f$ must be correct. 
By the condition in theorem~\ref{thm:bft-gps}, $P \rightarrow H$ where $H$ is a set of $f+1$ correct nodes. 
Due to quorum intersection, $S'\cap H$ must contain at least one node, which is correct. 
Let $c'$ be this node.
Since the graph is undirected, there exists $c \in S$ such that $c' \rightarrow c$. 

Let $t$ be the time $c'$ starts its vote timer. 
At time $t$, $c'$ also forwards the \msg{Propose}{v,val',-} message to all nodes. 
By time $t+d\Delta$, $c$ receives this message.
Thus, $c$ must have sent \msg{Vote-1}{v,val} before time $t+d\Delta$.
Otherwise, $c$ would have detected leader equivocation and would not have voted. 
Then, $c$ must have forwarded \msg{Propose}{v,val,-} to all nodes before time $t$.
$c'$ receives this \msg{Propose}{v,val,-} message before time $t+d\Delta$, which is before its vote timer expires. 
Thus, $c'$ detects leader equivocation and would not have voted. 
This contradicts $c' \in S'$. 
\end{proof}




\begin{lemma}\label{lem:hps-bft-agree-2}
If some node commits $val$ in view $v$, then any set of $n-f$ \msg{Vote-1}{v',val'} messages (lock certificate) in view $v'\geq v$ must have $val'=val$.
\end{lemma}

\begin{proof}
We prove this lemma by induction on view $v'$. 
The base case of $v'=v$ is straightforward by lemma \ref{lem:hps-bft-agree-1}. 
For the inductive step, suppose the lemma holds up to view $v'-1$, and now we consider view $v'$. Suppose for the sake of contradiction that some node commits $val$ in view $v$, and there exist $n-f>f$ nodes that send \msg{Vote-1}{v',val'} messages for $val' \neq val$.
A correct node will only send \msg{Vote-1}{v',val'} if a proposal carries in view $v'$ a set $S$ of \msg{Status}{v',-} messages. 
Thus, there exists a subset $H\subseteq S$ of $n-2f$ correct nodes which sent \msg{Status}{v',-}. 
By the condition in theorem~\ref{thm:bft-gps}, $H\rightarrow Q$, where $Q$ is a set of $f+1$ correct nodes.


Since a node committed $val$ in view $v$, there must exist some set $n-f$ nodes that sent \msg{Vote-2}{v, val}, of which a set $R$ of at least $n-2f$ are correct. 
Before sending \msg{Vote-2}{v,val} messages, these correct nodes updated $lock\coloneqq (v,val)$ in view $v$. Sets $Q$ and $R$ intersect in at least one correct node. Let this node be $q$. Since the graph is undirected and $H\rightarrow Q$, there must exist a node $h\in H$ such that $q\rightarrow h$. By the induction hypothesis, any lock certificate from view $v$ to $v'-1$ must be for $val$. Since a node only updates its lock monotonically based on view numbers, node $q$ must have a lock with view $\geq v$ for $val$. 
Let $t_h$ be the time node $h$ echoed $f+1$ \msg{ViewChange}{v'} messages. By time $t_h+d\Delta$, node $q$ must have received $f+1$ \msg{ViewChange}{v'} messages. Node $q$ will then echo a \msg{Locked}{lock} message to all nodes. This will be received by node $h$ by time $t_h+2d\Delta$. Node $h$ will update its lock to be at least view $v$ for $val$. 
Thus, from nodes in $H$, $L_{v'}$  receives at least one \textsc{Status} message for $val$ with view $\geq v$. By the induction assumption, any lock certificate not for $val$ must have view $<v$. 
Thus, no correct node sends \msg{Vote-1}{v',val'}, a contradiction.
\end{proof}



\begin{theorem}[Agreement]
No two correct nodes commit different values.    
\end{theorem}

\begin{proof}
Let $v$ be the smallest view in which a correct node commits some value, say $val$. 
By lemma~\ref{lem:hps-bft-agree-2}, only $val$ can receive $n-f$ \msg{Vote-1}{v} messages in any view $v' \geq v$, so no other value can be committed by a correct node.
\end{proof}



\begin{theorem}[Termination]
All correct nodes eventually decide.    
\end{theorem}

\begin{proof}
With round-robin leader election, correct nodes are elected leaders infinitely often. Thus, there must be a view $v$, after $GST+2d\Delta$, whose leader is correct. We next prove that all nodes will decide and terminate in view $v$ (if they don't decide earlier).

Let $t$ ($t \geq GST+2d\Delta$) be the first time some correct node enters view $v$. 
This correct node echoes $f+1$ \msg{ViewChange}{v-1} messages to all nodes at $t-2d\Delta \geq GST$. 
All correct nodes will receive the new view certificate by time $t-2d\Delta+\Delta$, wait $2d\Delta$ themselves, and enter view $v$ by time $t+\Delta$. 
Upon entering view $v$, they send \msg{Status}{v,-} messages to $L_v$. 
$L_v$ receives $n-f$ \msg{Status}{v,-} messages by time $t+2\Delta$, and send a \msg{Propose}{v,-} message to all nodes. 
All correct nodes will receive the \msg{Propose}{v,-} message by time $t+3\Delta$ and start their vote timers. 
Since $L_v$ is correct and does not equivocate, all correct nodes will send a \msg{Vote-1}{v,-} message by time $t+(3+d)\Delta$. 
All correct nodes will receive $n-f$ \msg{Vote-1}{v,-} messages by time $t+(4+d)\Delta$, and send a \msg{Vote-2}{v,-} message.
All correct nodes will receive $n-f$ \msg{Vote-2}{v,-} messages and commit by time $t+(5+d)\Delta)$. 
Since a node's view timer is $(5+d)\Delta$, and changing views requires $f+1$ \msg{ViewChange}{v} messages, all correct nodes will remain in view $v$, commit and terminate in view~$v$.
\end{proof}

\section{Related Work}

Necessary and sufficient conditions to solve consensus in all three classic timing models have been long established~\cite{byzantinegenerals,flm,dolevstrong,flp,bracha1987asynchronous,dls}.
There is also a large body of work on CFT
and BFT consensus protocols
in all three timing models. 
Our protocols adopt standard techniques from previous protocols such as quorum intersection~\cite{paxos,vr,castro1999pbft},
synchronous equivocation detection~\cite{katzkoo,synchs,good-case}, and view synchronizers~\cite{alexey-live,alexey-live-2}.


Weaker models than synchrony have been suggested in the literature. 
Some of these are orthogonal to the timing model.
A line of work studies consensus on \emph{incomplete} communication graphs~\cite{directed, undirected, undirected-asynchronous}. 
The mobile link failure model~\cite{mobilelinkfailure} allows a bounded number of lossy links.
These models are orthogonal because they still need to adopt one of the classic timing models for the links that exist in the graph and are not lossy. 
The mobile sluggish model~\cite{mobilesluggish} allows temporary unbounded message delays for a set of honest nodes (the set can change over time). 
The sleepy model~\cite{sleepy} allows a large fraction of nodes to be inactive. 
Both are models of node failures. 
Correct nodes that are not sluggish/sleepy are still assumed to have pair-wise synchronous links with each other.

The Visigoth fault tolerance (VFT) paper~\cite{visigoth} proposes a timing model that consists of synchronous and asynchronous links. 
Their model assumes every node has asynchronous links to at most $s$ correct nodes and synchronous links to the remaining nodes. For CFT, VFT requires $n-s\geq f+1$, so every node must have at least $f+1$ synchronous links. For BFT, VFT requires every node to have $n-s\geq 2f+1$ synchronous links. In contrast, our graph conditions are weaker (less restrictive) in that they only require a set of $n-f$ nodes for CFT ($n-2f$ correct nodes for BFT) to have synchronous paths to at least $f+1$ nodes ($f+1$ correct nodes for BFT). We additionally consider partially synchronous edges.


Another line of work that considers a mixture of links studies the minimal condition to circumvent the FLP~\cite{flp} impossibility and solve consensus deterministically~\cite{few-sync-links, little-sync, minimal-sync, necessary-byz, marcos-minimal}. 
Many of these works~\cite{few-sync-links, little-sync, minimal-sync} consider the harder setting of directed graphs, while we only consider undirected graphs. 
Since they focus on circumventing FLP, they only consider a mixture of asynchronous and partially synchronous links, but no synchronous links.
Our main focus is to use synchronous links to achieve better fault tolerance than those under partial synchrony.
But as mentioned, when $n>2f$ for crash and $n>3f$ for Byzantine, our ``safety-critical'' condition becomes vacuous, and our model degenerates to a mixture of partially synchronous and asynchronous links. 
In this context, our work establishes the minimum condition for circumventing FLP for CFT consensus in undirected graphs.

\section{Conclusion}
This paper introduces the \sys{} model that considers a mixture of synchronous, partially synchronous, and asynchronous links to better capture the heterogeneity of modern networks. We present necessary and sufficient conditions for solving crash and Byzantine consensus in \sys. Our results show that consensus is solvable in the presence of $f \geq n/2$ crash faults and $f \geq n/3$ Byzantine faults in \sys, even though not all links are synchronous.

\paragraph{Acknowledgement.}
This work was started while Ittai Abraham, Neil Giridharan and Ling Ren were at VMware Research. This work is funded in part by the National Science Foundation award \#2143058.

\bibliographystyle{plainurl}
\bibliography{references}

\begin{thebibliography}{10}

\bibitem{synchs}
Ittai Abraham, Dahlia Malkhi, Kartik Nayak, Ling Ren, and Maofan Yin.
\newblock Sync hotstuff: Simple and practical synchronous state machine replication.
\newblock In {\em 2020 IEEE Symposium on Security and Privacy (SP)}, pages 106--118, 2020.
\newblock \href {https://doi.org/10.1109/SP40000.2020.00044} {\path{doi:10.1109/SP40000.2020.00044}}.

\bibitem{good-case}
Ittai Abraham, Kartik Nayak, Ling Ren, and Zhuolun Xiang.
\newblock Good-case latency of byzantine broadcast: A complete categorization.
\newblock In {\em Proceedings of the 2021 ACM Symposium on Principles of Distributed Computing}, PODC'21, page 331–341, New York, NY, USA, 2021. Association for Computing Machinery.
\newblock \href {https://doi.org/10.1145/3465084.3467899} {\path{doi:10.1145/3465084.3467899}}.

\bibitem{saksham}
Saksham Agarwal, Qizhe Cai, Rachit Agarwal, David Shmoys, and Amin Vahdat.
\newblock Harmony: A congestion-free datacenter architecture.
\newblock In {\em 21st USENIX Symposium on Networked Systems Design and Implementation (NSDI 24)}, pages 329--343, Santa Clara, CA, April 2024. USENIX Association.
\newblock URL: \url{https://www.usenix.org/conference/nsdi24/presentation/agarwal-saksham}.

\bibitem{marcos-minimal}
Marcos~K. Aguilera, Carole Delporte-Gallet, Hugues Fauconnier, and Sam Toueg.
\newblock Communication-efficient leader election and consensus with limited link synchrony.
\newblock In {\em Proceedings of the Twenty-Third Annual ACM Symposium on Principles of Distributed Computing}, PODC '04, page 328–337, New York, NY, USA, 2004. Association for Computing Machinery.
\newblock \href {https://doi.org/10.1145/1011767.1011816} {\path{doi:10.1145/1011767.1011816}}.

\bibitem{little-sync}
M.K. Aguilera, C.~Delporte-Gallet, H.~Fauconnier, and S.~Toueg.
\newblock Consensus with byzantine failures and little system synchrony.
\newblock In {\em International Conference on Dependable Systems and Networks (DSN'06)}, pages 147--155, 2006.
\newblock \href {https://doi.org/10.1109/DSN.2006.22} {\path{doi:10.1109/DSN.2006.22}}.

\bibitem{wpaxos}
Ailidani Ailijiang, Aleksey Charapko, Murat Demirbas, and Tevfik Kosar.
\newblock Wpaxos: Wide area network flexible consensus.
\newblock {\em IEEE Trans. Parallel Distrib. Syst.}, 31(1):211–223, jan 2020.
\newblock \href {https://doi.org/10.1109/TPDS.2019.2929793} {\path{doi:10.1109/TPDS.2019.2929793}}.

\bibitem{necessary-byz}
Olivier Baldellon, Achour Most{\'e}faoui, and Michel Raynal.
\newblock A necessary and sufficient synchrony condition for solving byzantine consensus in symmetric networks.
\newblock In {\em International Conference on Distributed Computing and Networking}, pages 215--226. Springer, 2011.

\bibitem{minimal-sync}
Zohir Bouzid, Achour Mostfaoui, and Michel Raynal.
\newblock Minimal synchrony for byzantine consensus.
\newblock In {\em Proceedings of the 2015 ACM Symposium on Principles of Distributed Computing}, pages 461--470, 2015.

\bibitem{bracha1987asynchronous}
Gabriel Bracha.
\newblock Asynchronous byzantine agreement protocols.
\newblock {\em Information and Computation}, 75(2):130--143, 1987.

\bibitem{alexey-live-2}
Manuel Bravo, Gregory Chockler, and Alexey Gotsman.
\newblock {Liveness and Latency of Byzantine State-Machine Replication}.
\newblock In Christian Scheideler, editor, {\em 36th International Symposium on Distributed Computing (DISC 2022)}, volume 246 of {\em Leibniz International Proceedings in Informatics (LIPIcs)}, pages 12:1--12:19, Dagstuhl, Germany, 2022. Schloss Dagstuhl -- Leibniz-Zentrum f{\"u}r Informatik.
\newblock URL: \url{https://drops.dagstuhl.de/entities/document/10.4230/LIPIcs.DISC.2022.12}, \href {https://doi.org/10.4230/LIPIcs.DISC.2022.12} {\path{doi:10.4230/LIPIcs.DISC.2022.12}}.

\bibitem{alexey-live}
Manuel Bravo, Gregory Chockler, and Alexey Gotsman.
\newblock Making byzantine consensus live.
\newblock {\em Distributed Computing}, 35(6):503--532, 2022.

\bibitem{cachin-aba}
Christian Cachin, Klaus Kursawe, and Victor Shoup.
\newblock Random oracles in constantipole: practical asynchronous byzantine agreement using cryptography (extended abstract).
\newblock In {\em Proceedings of the Nineteenth Annual ACM Symposium on Principles of Distributed Computing}, PODC '00, page 123–132, New York, NY, USA, 2000. Association for Computing Machinery.
\newblock \href {https://doi.org/10.1145/343477.343531} {\path{doi:10.1145/343477.343531}}.

\bibitem{castro1999pbft}
Miguel Castro and Barbara Liskov.
\newblock Practical byzantine fault tolerance.
\newblock In {\em Proceedings of the Third Symposium on Operating Systems Design and Implementation}, OSDI '99, page 173–186, USA, 1999. USENIX Association.

\bibitem{bgp-hijack}
Shinyoung Cho, Romain Fontugne, Kenjiro Cho, Alberto Dainotti, and Phillipa Gill.
\newblock Bgp hijacking classification.
\newblock In {\em 2019 Network Traffic Measurement and Analysis Conference (TMA)}, pages 25--32, 2019.
\newblock \href {https://doi.org/10.23919/TMA.2019.8784511} {\path{doi:10.23919/TMA.2019.8784511}}.

\bibitem{dolevstrong}
D.~Dolev and H.~R. Strong.
\newblock Authenticated algorithms for byzantine agreement.
\newblock {\em SIAM J. Comput.}, 12(4):656–666, nov 1983.
\newblock \href {https://doi.org/10.1137/0212045} {\path{doi:10.1137/0212045}}.

\bibitem{dls}
Cynthia Dwork, Nancy Lynch, and Larry Stockmeyer.
\newblock Consensus in the presence of partial synchrony.
\newblock {\em J. ACM}, 35(2):288–323, apr 1988.
\newblock \href {https://doi.org/10.1145/42282.42283} {\path{doi:10.1145/42282.42283}}.

\bibitem{flm}
Michael~J. Fischer, Nancy~A. Lynch, and Michael Merritt.
\newblock Easy impossibility proofs for distributed consensus problems.
\newblock In {\em Proceedings of the Fourth Annual ACM Symposium on Principles of Distributed Computing}, PODC '85, page 59–70, New York, NY, USA, 1985. Association for Computing Machinery.
\newblock \href {https://doi.org/10.1145/323596.323602} {\path{doi:10.1145/323596.323602}}.

\bibitem{flp}
Michael~J. Fischer, Nancy~A. Lynch, and Michael~S. Paterson.
\newblock Impossibility of distributed consensus with one faulty process.
\newblock {\em J. ACM}, 32(2):374–382, apr 1985.
\newblock \href {https://doi.org/10.1145/3149.214121} {\path{doi:10.1145/3149.214121}}.

\bibitem{mobilesluggish}
Yue Guo, Rafael Pass, and Elaine Shi.
\newblock Synchronous, with a chance of partition tolerance.
\newblock In {\em Advances in Cryptology – CRYPTO 2019: 39th Annual International Cryptology Conference, Santa Barbara, CA, USA, August 18–22, 2019, Proceedings, Part I}, page 499–529, Berlin, Heidelberg, 2019. Springer-Verlag.
\newblock \href {https://doi.org/10.1007/978-3-030-26948-7_18} {\path{doi:10.1007/978-3-030-26948-7_18}}.

\bibitem{few-sync-links}
Moumen Hamouma, Achour Most{\'e}faoui, and Gilles Tr{\'e}dan.
\newblock Byzantine consensus with few synchronous links.
\newblock In {\em Principles of Distributed Systems: 11th International Conference, OPODIS 2007, Guadeloupe, French West Indies, December 17-20, 2007. Proceedings 11}, pages 76--89. Springer, 2007.

\bibitem{cloudy}
Owen Hilyard, Bocheng Cui, Marielle Webster, Abishek~Bangalore Muralikrishna, and Aleksey Charapko.
\newblock Cloudy forecast: How predictable is communication latency in the cloud?, 2023.
\newblock \href {https://arxiv.org/abs/2309.13169} {\path{arXiv:2309.13169}}.

\bibitem{latency-variation-internet}
Toke H\o{}iland-J\o{}rgensen, Bengt Ahlgren, Per Hurtig, and Anna Brunstrom.
\newblock Measuring latency variation in the internet.
\newblock In {\em Proceedings of the 12th International on Conference on Emerging Networking EXperiments and Technologies}, CoNEXT '16, page 473–480, New York, NY, USA, 2016. Association for Computing Machinery.
\newblock \href {https://doi.org/10.1145/2999572.2999603} {\path{doi:10.1145/2999572.2999603}}.

\bibitem{katzkoo}
Jonathan Katz and Chiu-Yuen Koo.
\newblock On expected constant-round protocols for byzantine agreement.
\newblock In Cynthia Dwork, editor, {\em Advances in Cryptology - CRYPTO 2006}, pages 445--462, Berlin, Heidelberg, 2006. Springer Berlin Heidelberg.

\bibitem{undirected}
Muhammad~Samir Khan, Syed~Shalan Naqvi, and Nitin~H. Vaidya.
\newblock Exact byzantine consensus on undirected graphs under local broadcast model, 2019.
\newblock \href {https://arxiv.org/abs/1903.11677} {\path{arXiv:1903.11677}}.

\bibitem{undirected-asynchronous}
Muhammad~Samir Khan and Nitin Vaidya.
\newblock Asynchronous byzantine consensus on undirected graphs under local broadcast model, 2019.
\newblock \href {https://arxiv.org/abs/1909.02865} {\path{arXiv:1909.02865}}.

\bibitem{paxos}
Leslie Lamport.
\newblock The part-time parliament.
\newblock {\em ACM Trans. Comput. Syst.}, 16(2):133–169, may 1998.
\newblock \href {https://doi.org/10.1145/279227.279229} {\path{doi:10.1145/279227.279229}}.

\bibitem{byzantinegenerals}
Leslie Lamport, Robert Shostak, and Marshall Pease.
\newblock The byzantine generals problem.
\newblock {\em ACM Trans. Program. Lang. Syst.}, 4(3):382–401, jul 1982.
\newblock \href {https://doi.org/10.1145/357172.357176} {\path{doi:10.1145/357172.357176}}.

\bibitem{epaxos}
Iulian Moraru, David~G. Andersen, and Michael Kaminsky.
\newblock There is more consensus in egalitarian parliaments.
\newblock In {\em Proceedings of the Twenty-Fourth ACM Symposium on Operating Systems Principles}, SOSP '13, page 358–372, New York, NY, USA, 2013. Association for Computing Machinery.
\newblock \href {https://doi.org/10.1145/2517349.2517350} {\path{doi:10.1145/2517349.2517350}}.

\bibitem{vr}
Brian~M. Oki and Barbara~H. Liskov.
\newblock Viewstamped replication: A new primary copy method to support highly-available distributed systems.
\newblock In {\em Proceedings of the Seventh Annual ACM Symposium on Principles of Distributed Computing}, PODC '88, page 8–17, New York, NY, USA, 1988. Association for Computing Machinery.
\newblock \href {https://doi.org/10.1145/62546.62549} {\path{doi:10.1145/62546.62549}}.

\bibitem{sleepy}
Rafael Pass and Elaine Shi.
\newblock The sleepy model of consensus.
\newblock In Tsuyoshi Takagi and Thomas Peyrin, editors, {\em Advances in Cryptology -- ASIACRYPT 2017}, pages 380--409, Cham, 2017. Springer International Publishing.

\bibitem{visigoth}
Daniel Porto, Jo\~{a}o Leit\~{a}o, Cheng Li, Allen Clement, Aniket Kate, Flavio Junqueira, and Rodrigo Rodrigues.
\newblock Visigoth fault tolerance.
\newblock In {\em Proceedings of the Tenth European Conference on Computer Systems}, EuroSys '15, New York, NY, USA, 2015. Association for Computing Machinery.
\newblock \href {https://doi.org/10.1145/2741948.2741979} {\path{doi:10.1145/2741948.2741979}}.

\bibitem{swift-paxos}
Fedor Ryabinin, Alexey Gotsman, and Pierre Sutra.
\newblock {SwiftPaxos}: Fast {Geo-Replicated} state machines.
\newblock In {\em 21st USENIX Symposium on Networked Systems Design and Implementation (NSDI 24)}, pages 345--369, Santa Clara, CA, April 2024. USENIX Association.
\newblock URL: \url{https://www.usenix.org/conference/nsdi24/presentation/ryabinin}.

\bibitem{mobilelinkfailure}
U.~Schmid, B.~Weiss, and J.~Rushby.
\newblock Formally verified byzantine agreement in presence of link faults.
\newblock In {\em Proceedings 22nd International Conference on Distributed Computing Systems}, pages 608--616, 2002.
\newblock \href {https://doi.org/10.1109/ICDCS.2002.1022311} {\path{doi:10.1109/ICDCS.2002.1022311}}.

\bibitem{directed}
Lewis Tseng and Nitin Vaidya.
\newblock Exact byzantine consensus in directed graphs, 2014.
\newblock \href {https://arxiv.org/abs/1208.5075} {\path{arXiv:1208.5075}}.

\bibitem{synchronous-data-center}
Tian Yang, Robert Gifford, Andreas Haeberlen, and Linh Thi~Xuan Phan.
\newblock The synchronous data center.
\newblock In {\em Proceedings of the Workshop on Hot Topics in Operating Systems}, HotOS '19, page 142–148, New York, NY, USA, 2019. Association for Computing Machinery.
\newblock \href {https://doi.org/10.1145/3317550.3321442} {\path{doi:10.1145/3317550.3321442}}.

\end{thebibliography}

\appendix
\section{BFT Consensus in \gascaps}
We present a sufficient condition for solving BFT consensus in \gas. 

\begin{theorem}\label{thm:bft-gpa}
If (i) the condition in theorem~\ref{thm:bft-gps} holds and 
(ii) for all $F$ with $|F|=f$, there exists a node in graph $G'=(V-F,\diamond E)$, which has partially synchronous paths to $f$ other nodes in $G'$,
then BFT consensus on graph $G=(V,E)$ can be solved deterministically under \gas. 
\end{theorem}



We have proved the necessity of condition (i) (for all algorithms) in Section~\ref{s:hps-bft-necessary}.
Condition (ii) was proven necessary in~\cite{necessary-byz} for algorithms that work for the family of all graphs that satisfy the condition (i.e., graph-agnostic algorithms).
If algorithms can be tailored to the graph, the tight condition for Byzantine consensus remains open.

\subsection{Protocol}

Next, we adapt our previous BFT consensus protocol in algorithm~\ref{alg:gps-bft} from \gps{} to \gas{}, assuming the condition in theorem~\ref{thm:bft-gpa} holds. This establishes the sufficiency of the condition.

As with our \gas{} CFT consensus protocol, we utilize condition (ii) in theorem~\ref{thm:cft-gpa} to guarantee that, when the correct node with partially synchronous paths to $f$ other nodes in $G'=(V-F,\diamond E)$ is elected leader after GST, this leader will not be replaced and will make progress. 
To do so, we first require $n-f$ nodes to initiate a view change, of which at least $n-2f$ must be correct.
This way, because fewer than $n-2f$ correct nodes are asynchronously connected to the leader, we just need to make sure that none of the $f$ nodes the leader is connected to via partially synchronous paths initiates a view change. 



\begin{algorithm}[tb]
\caption{BFT consensus protocol in \gas{} for node $i$}\label{alg:gas-bft}
\begin{algorithmic}[1]
\State{$v_i \gets 0,~lock \gets \bot$}\Comment{Initialize local view number and lock}
\State{\textbf{enter} view 1}    

\medskip
\Event{entering view $v$}
\State{$v_i \gets v$}
\State{\textbf{send} \msg{Status}{v,lock} to all}
\EndEvent

\medskip
\Event{receiving $n-f$ \msg{Status}{v_i,-} messages where $i\neq L_v$}
\State{\textbf{echo} these $n-f$ \msg{Status}{v_i,-} to all}
\State{\textbf{start} $proposal\_timer \gets timer(3d'\Delta)$}\Comment{Timer before changing view}
\EndEvent

\medskip
\Event{$proposal\_timer$ expiring and no leader proposal received}
\State{\textbf{send} \msg{ViewChange}{v_i} to all}
\EndEvent

\medskip
\State{\textsc{Proposal, Vote-1, Vote-2, Commit} messages, $n-f$ \textsc{ViewChange} messages (instead of $f+1$), equivocation detection at all nodes, and \textsc{Status} messages at view leaders are processed the same way as in Algorithm~\ref{alg:gps-bft}}

\end{algorithmic}
\end{algorithm}

We only describe the status and view change steps, since the rest of the protocol remains the same as algorithm~\ref{alg:gps-bft}.

\par \underline{\textit{Status step.}}
Upon entering a new view $v$, a node sends a \msg{Status}{v,lock} message to \textbf{all} nodes. When a node receives at least $n-f$ \msg{Status}{v,-} messages, it forwards this set of \textsc{Status} messages to all nodes and starts a timer with $3d'\Delta$ duration. The same propose, vote, and commit steps from algorithm~\ref{alg:gps-bft} follow.

\par \underline{\textit{View change.}} A node suspects the leader is faulty if it does not receive a \msg{Propose}{v,-,-} message before its proposal timer (instead of view timer) expires. 
A view change certificate consists of $n-f$ \msg{ViewChange}{v} messages (instead of $f+1$ in algorithm~\ref{alg:gps-bft}). Upon receiving $n-f$ \msg{ViewChange}{v}, a node carries out the same waiting period step from algorithm \ref{alg:gps-cft}.

\subsection{Analysis}

The agreement and validity proofs are identical to the \gps{} BFT case. We focus on termination. 

\begin{lemma}
If no correct node ever terminates, then every correct node keeps entering higher views.
\label{lemma:bft-gas-liveness}
\end{lemma}

\begin{proof}
Suppose for the sake of contradiction, there exists a correct node $n_1$, which never enters a higher view. Let $v$ be the view $n_1$ is in.
If any correct node ever enters a view $v'>v$, it must have echoed $n-f$ \msg{ViewChange}{v'-1} messages to all nodes. $n_1$ will eventually receive this set $n-f$ \msg{ViewChange}{v'-1} messages and enter a higher view, a contradiction. 
Thus, no correct node ever enters a view higher than $v$.
Before entering view $v$, $n_1$ must have sent $n-f$ \msg{ViewChange}{v-1} to all nodes. All correct nodes will eventually receive this set of \msg{ViewChange}{v-1} messages, enter view $v$, and send a \msg{Status}{v,-} message. Eventually, correct nodes will receive $n-f$ \msg{Status}{v,-} messages and start their proposal timers. If $n_1$ receives $n-f$ \msg{ViewChange}{v} messages, it will enter view $v+1$, a contradiction. Thus $n_1$ never receives $n-f$ \msg{ViewChange}{v} messages. 
Then, there must be at least one correct node, $n_2$, which never sends \msg{ViewChange}{v}, and instead echoes \msg{Propose}{v,-,-} to all nodes. Eventually, all correct nodes will receive a \msg{Propose}{v,-,-} message and echo it. If a correct node detects leader equivocation, it will forward it to all correct nodes. $n_2$ will eventually receive the conflicting \textsc{Propose} messages and send a \msg{ViewChange}{v} message, a contradiction. Thus, no correct node will detect leader equivocation. Then, all correct nodes will send \msg{Vote-1}{v,-} messages to all nodes. Eventually all correct nodes will receive $n-f$ \msg{Vote-1}{v,-} messages, and send a \msg{Vote-2}{v,-} message. Eventually, $n_1$ will receive $n-f$ \msg{Vote-2}{v,-} messages, commit and terminate, a contradiction.
\end{proof}

\begin{theorem}
All correct nodes eventually terminate.    
\end{theorem}

\begin{proof}
Suppose for the sake of contradiction that some correct node never terminates.
Observe that if one correct node terminates, it sends a \textsc{Commit} message and makes all correct nodes eventually terminate.  
Thus, no correct node ever terminates. 
By lemma~\ref{lemma:bft-gas-liveness}, every correct node keeps entering higher views.

Eventually, there will be a first time after $GST+2d\Delta$ that some correct node enters a view $v$ such that (i) $L_v \not\in F$, and (ii) $L_v$ has paths to at least $f$ other nodes in graph $G'=(V-F,\diamond E)$.
Let $C$ denote this set of nodes including $L_v$. 
We next prove no node in $C$ will ever send \msg{ViewChange}{v}.

Let $p$ be the first node in $C$ that enters view $v$, and let $p$ enter view $v$ at time $t > GST+2d\Delta$. 
Observe that no node in $C$ will send \msg{ViewChange}{v} before time $t+3d'\Delta$ (proposal timer duration is $3d'\Delta$).
Due to the condition in theorem~\ref{thm:bft-gpa}, $n-|C|<n-f$.
Thus, there will not be $n-f$ \msg{ViewChange}{v} messages before $t+3d'\Delta$. 

$p$ sends $n-f$ \msg{ViewChange}{v-1} messages at time $t-2d\Delta>GST$. 
All nodes in $C$ receive $n-f$ \msg{ViewChange}{v-1} messages by time $t-2d\Delta+d'\Delta$, enter view $v$ by time $t+d'\Delta$, and stay in view $v$ at least until time $t+3d'\Delta$.

When a node $q \in C$ receives $n-f$ \msg{Status}{v,-} messages at time $t'>t$, $q$ echoes these $n-f$ messages and starts its proposal timer. 
All nodes in $C$ enter view $v$ by time $t+d'\Delta$ and are ready to echo these \msg{Status}{v,-} messages by $t+d'\Delta$.
$L_v$, which is in $C$, receives these $n-f$ \msg{Status}{v,-} messages by time $\max(t+2d'\Delta, t'+d'\Delta)<t'+2d'\Delta$.
$L_v$ sends a \msg{Propose}{v,-} message by time $t'+2d'\Delta$ and it reaches $q$ by time $t'+3d'\Delta$, which is before $q$'s proposal timer expires. 
Thus, $q$ does not send \msg{ViewChange}{v}. 
This establishes that no node in $C$ will ever send \msg{ViewChange}{v}.

Since $n-|C|<n-f$, there will never be $n-f$ \msg{ViewChange}{v} messages.
Thus, no correct node ever enters a view higher than $v$.
This contradicts lemma~\ref{lemma:bft-gas-liveness}.
\end{proof}
\section{Comparison with~\cite{marcos-minimal}}\label{s:marcos-minimal}

\cite{marcos-minimal} showed that a correct $\diamond f$-source is a sufficient condition for solving CFT consensus in a directed graph. A correct $\diamond f$-source is a correct node that has $f$ outgoing  
fault-free paths that are eventually synchronous.
\cite{marcos-minimal} argued the potential optimality of their result by showing that every node being a $\diamond(f-1)$-source is not sufficient for solving CFT consensus.
Our results show that, at least in the case of undirected graphs, a correct $\diamond f$-source is not necessary. 
Our condition (ii) in theorem \ref{thm:cft-gpa} is weaker and is sufficient.

To show our condition is weaker, we first prove that a correct $\diamond f$-source implies the condition (ii) in theorem \ref{thm:cft-gpa}. 
Let $C$ be the connected component in $G'=(V,\diamond E)$ that the correct $\diamond f$-source belongs to. 
We have $|C|\geq f+1$. 
Removing $F\cup C$ must leave at most $n-f-1$ nodes in the remaining graph. 

\begin{figure}
    \centering
    \includegraphics[width=6cm,height=9cm,keepaspectratio]{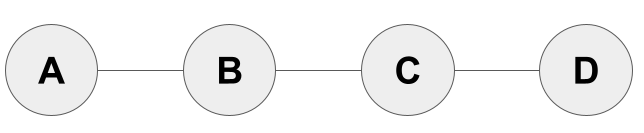}
    \qquad

    \caption{In this graph $n=4$ and $f=2$. Each edge represents a synchronous link and a missing edge represents an asynchronous link.}
    \label{fig:cft-weaker-figure}
\end{figure}

Next, Figure \ref{fig:cft-weaker-figure} shows an example of a graph that satisfies our condition but does not have a correct $\diamond f$-source. For this graph, if the adversary corrupts $B$ and $C$, then there is no correct $\diamond f$-source since $A$ only has a link to $B$ and $D$ only has a link to $C$. This graph, however, satisfies the condition (ii) in theorem \ref{thm:cft-gpa}. If $|F|=0$, removing the largest connected component (the entire graph) leaves 0 nodes, satisfying the condition. For any choice of $F$ with $|F|=1$, the largest connected component after removing $F$ must be of size at least $2$. Thus, there will be at most $1$ remaining node, satisfying the condition.
For any choice of $F$ such that $|F|=2$, the largest remaining connected component must be of size at least $1$. Thus, there will be at most $1$ remaining node, satisfying the condition.
\section{BFT Unanimity Validity}
\label{s:unanimity}

\begin{algorithm}[tb]
\caption{BFT Unanimity Validity}\label{alg:bft-unanimity-validity}
\begin{algorithmic}[1]
\State{$v_i \gets 0,~lock \gets \bot$}\Comment{Initialize local view number and lock}
\State{$inputs\gets \{\}$}
\medskip
\State{\textbf{echo} \msg{Input}{input_i} to all}
\State{\textbf{start} $input\_timer \gets timer(2d\Delta)$}

\medskip
\Event{receiving $m\gets$ \msg{Input}{input_j}}
\State{\textbf{echo} $m$}
\State{$inputs\gets inputs \cup \{m\}$}
\EndEvent

\medskip
\Event{$input\_timer$ expiring}
\State{\textbf{send} \msg{Forward-Inputs}{inputs} to all}
\EndEvent

\medskip
\Event{receiving $FI \gets n-f$ \msg{Forward-Inputs}{inputs}}
\If{having received $I\gets f+1$ \msg{Input}{val} messages in $FI$}
\State{$lock\gets I$}
\EndIf
\EndEvent

\medskip
\State{\textbf{enter} view $1$}


\medskip

\end{algorithmic}
\end{algorithm}

In this section, we give a way to convert our BFT algorithms from external validity to strong unanimity validity. 
The idea is to try to have nodes lock before starting the first view, and if all correct nods have the same input, then that input is the only lock. 


\begin{lemma}\label{bft-validity-1}
If all correct nodes have the same input, then all correct nodes will lock on this value before entering view $1$, and any lock in view $0$ must be for $val$.
\end{lemma}

\begin{proof}
In view $0$, all correct nodes send their inputs and echo other nodes' inputs they receive (using \textsc{Input} and \textsc{Forward-Inputs} messages) before their input timer expires in $2d\Delta$ time. For any two correct nodes $p$ and $q$ such that $p\rightarrow q$, $p$ will receive $q$'s input before $p$'s input timer expires. Similarly, $q$ will receive $p$'s input before $q$'s input timer expires. Consider any correct node $c$. Node $c$ will eventually receive a set $A$ of $n-f$ \msg{Forward-Inputs}{inputs} messages. 
Among them, a subset $B$ of $n-2f$ are from correct nodes. By the condition in theorem \ref{thm:bft-gps}, $B\rightarrow C$ where $C$ is a set of $f+1$ correct nodes. Since every node in $B$ waits $2d\Delta$ before sending a \textsc{Forward-Inputs} message, this is sufficient time for each node in $C$ to receive an input from some node in $B$ and also sends its input to that node in $B$.
Thus, $B$ will contain the input values from $C$, a set of $f+1$ correct nodes. If all correct nodes have the input $val$, node $c$ must receive at least $f+1$ \msg{Input}{val} messages, and there are at most $f$ \textsc{Input} messages for a different value (from $f$ Byzantine nodes).
Therefore, every correct node will set its lock to $I\gets f+1$ \msg{Input}{val} in view $0$, and any lock in view 0 must be for $val$.
\end{proof}

\begin{lemma}\label{lem:bft-validity-2}
If all correct nodes have the same input, then any lock in view $v\geq 0$ must be for $val$.
\end{lemma}

\begin{proof}
The base case is established by lemma~\ref{bft-validity-1}.
Now assume the lemma holds for all $v-1$, and consider view $v$. Suppose for the sake of contradiction a lock forms for $val'\neq val$. $L_v$ must have proposed $val'\neq val$. By the induction assumption, any lock must be for $val$. Thus, $L_v$ must have received $S\gets n-f$ \textsc{Status} messages where all locks are $\bot$. By lemma \ref{bft-validity-1}, all correct nodes will lock on $val$ before entering view $1$. The set $S$ must contain a \textsc{Status} message from at least one correct node. This correct node will at least have a lock in view $0$ or higher, and thus its \textsc{Status} message will not have $lock=\bot$, a contradiction.
\end{proof}

\begin{theorem}
If all correct nodes have the same input, then only that value can be decided.  
\end{theorem}

\begin{proof}
By lemma \ref{lem:bft-validity-2}, any lock must be for $val$, the input of the correct nodes. Only locked values can be decided. Validity then follows from termination.
\end{proof}

\end{document}